\documentclass[journal]{IEEEtran}

\usepackage{graphicx}
\usepackage{url}
\usepackage{epstopdf}
\usepackage{amssymb}
\usepackage{amsmath}
\usepackage{bm}
\usepackage{amsthm}
\newtheorem{thm}{Theorem}[section]
\newtheorem{lem}[thm]{Lemma}
\usepackage[numbers,square, sort&compress]{natbib}
\usepackage{mathtools}

\DeclarePairedDelimiter\floor{\lfloor}{\rfloor}
\usepackage{url}

\begin{document}
%
\title{Cost Effective Rumor Containment in Social Networks}

\author{\IEEEauthorblockN{Bhushan Kotnis and Joy Kuri}

\IEEEauthorblockA{
Dept. of Electronic Systems Engineering, 
Indian Institute of Science, 
Bangalore.  
\{bkotnis,kuri\}@dese.iisc.ernet.in
}
}
\maketitle

\begin{abstract}
The spread of rumors through social media and online social networks can not only disrupt the daily lives of citizens, but also result in loss of life and property. A rumor spreads when individuals, who are unable decide the authenticity of the information, mistake the rumor as genuine information and pass it to their acquaintances. We propose a solution where a set of individuals, characterized by their degree in the social network, are trained and provided resources to help them distinguish a rumor from genuine information. By formulating an optimization problem we calculate the set of individuals, who must undergo training, and the quality of training that minimizes the expected training cost and ensures an  upper bound on the size of the rumor outbreak. Our primary contribution is, that although the optimization problem turns out to be non linear, we show that the problem is equivalent to solving a set of linear programs. This result also allows us to solve the problem of minimizing the size of rumor outbreak for a given cost budget. The solution of the optimization problem demonstrates  that increasing the number of trained individuals is better than improving the quality of training. Furthermore, the optimum set of trained individuals displays a pattern which can be converted to an implementable heuristic. These results can prove to be very useful to social planners and law enforcement agencies for preventing dangerous rumors and misinformation epidemics.
\end{abstract}

\IEEEpeerreviewmaketitle

\section{Introduction}
The past decade has seen a dramatic increase in the usage of online social networking (OSN) and microblogging services \cite{Java2007} like Facebook, Google+, Twitter, etc. Apart from their usefulness in helping individuals keep in touch, they are increasingly being used for disseminating information about events happing in real time \cite{Leskovec2009}. Due to the proliferation of smart phones, individuals can easily capture the unfolding of events in real time and can share it with others via social media instantaneously. Social media played a key role in sharing of information in the immediate aftermath of: the Fukushima nuclear accident \cite{Friedman2011}, hurricane Sandy \cite{Gupta2013} and the Boston marathon bombings \cite{Cassa2013}. Online social networks, microblogging, and short messaging services can be extremely useful in aiding the dissemination of time critical and potentially life saving information during large scale human emergencies. 
\par
However, along with useful information, OSN  can also aid the spread of rumors, and can even facilitate the spread of a dangerous misinformation epidemic \cite{Gupta2013}. Since these services operate in a decentralized manner, i.e., absence of a central authority for guaranteeing the authenticity of the information, careless individuals can accidentally initiate and propagate rumors. Furthermore, the distributed nature of social media can be exploited by malicious agents for spreading dangerous misinformation epidemics. This problem has been a growing concern for administrative authorities \cite{chineserumor}. In the recent past, with the goal of preventing the spread of rumors, the Indian Government ordered cell phone operators to impose a limit on the number of text messages that can be sent by any individual \cite{Franke2012}. However, such ad-hoc measures are costly and also not very effective. In this article we propose a cost effective rumor prevention mechanism which provides guarantees on the size of the rumor outbreak.
\par
The problem of maximizing the spread of information in social networks is well known \cite{Kempe2003,Chen2009,Chen2010}. Approaches such as calculating the optimal (or approximately optimal) set of seed nodes \cite{Kempe2003, Kempe2005,Chen2009} and optimal control \cite{Karnik2012,Dayama2012,Kandhway2014} have been proposed. The reverse problem of limiting the spread of rumors has also received some attention and is termed as least cost rumor blocking problem  \cite{Fan2013}, misinformation containment problem  \cite{Nguyen2012}, influence blocking maximization problem  \cite{He2012}. The predominant approach used to address this problem is to identify a set of target seed nodes who will spread an `anti-rumor' \cite{Tripathy2010,Budak2011,Nguyen2012, He2012,Fan2013} to combat the rumor. The problem of limiting the spread of virus or malware in computer networks is analogous to the problem of limiting the spread of rumors, and studies \cite{Khouzani2010,Khouzani2010b} have used optimal control techniques for achieving the same. However, the optimal control solution is difficult to implement as it requires a centralized real time controller. 
\par
Our approach, which involves training nodes to help them distinguish between the rumor and true information, is similar to link blocking in \cite{Kimura2008}  and node removal strategy in\cite{Cohen2003, Aspnes2005,Habiba2010}.  However, apart from \cite{Aspnes2005} most studies do not consider the cost of blocking links or nodes.  The problem of minimizing the cost of immunizing nodes for preventing the spread of a disease (or rumor) was studied in  \cite{Aspnes2005} and it was shown to be an NP-Hard  problem. Their model assumed perfect immunization (immunized node cannot transmit a rumor), while our model allows for partial immunization (immunized node can transmit a rumor). In this study we formulate an optimization problem in which, apart from the set of nodes to be immunized, the immunization quality is also a control variable, and which also incorporates cost budget constraints and constraint on the number of immunized nodes. Thus, the problem we are considering is more richer and general than the one in \cite{Aspnes2005}.
\par
The straightforward approach of calculating the set of nodes which minimizes the cost and ensures prevention of rumor outbreak suffers from a combinatorial explosion problem. For the purpose of analytical tractability we characterize nodes solely based on their degree.  Although a degree centrality based approach may not be optimal, it has been shown to work very well in scale free networks \cite{Albert2000, Callaway2000, Cohen2003,Habiba2010}. Furthermore, a solution based on the degree based characterization of nodes is very simple to implement.
\par  
We propose a method in which a set of individuals, characterized by their degree, are trained to distinguish rumors from useful information. Such training may also involve allocating costly resources, such as access to security cameras, drone or satellite photography, or access to a group of experts which help individuals make the correct decision. Thus, training individuals is a costly affair. We seek to identify the set of individuals to be trained, that minimizes the expected training cost and prevents a rumor outbreak by formulating and solving an optimization problem. We also address the problem of minimizing the size of the rumor outbreak for a given cost budget.
\par
 Our model assumes two types of individuals viz. untrained and trained, and they are connected to one another through a social network.  Both, the untrained and the trained are unable to distinguish between the rumor with certainty, however trained individuals can distinguish a rumor from an information better then the untrained ones.  Individuals are characterized by their degree, and the cost of training is assumed to be proportional to their degree. The \emph{Independent Cascade}  (IC) model discussed in \cite{Kempe2003} is used to model the rumor propagation.
 \par
  For a given set of trained individuals, we first calculate the size of the rumor outbreak using network percolation theory, and then find the set of nodes which 1. minimizes the cost for a given outbreak size and 2. minimizes outbreak size for a given cost budget. Both the optimization problems are found to be a non linear. Our primary contribution is, that we show that the nonlinear problem can be addressed by solving a set of linear programs.  Furthermore, we discover that the set of individuals which need to be trained displays an interesting pattern: it turns out that when, the cost is linear in terms of node degree, low degree individuals are more important than high degree individuals for the purpose of rumor prevention.
\par
Our contributions are summarized as follows :
\begin{itemize}
\item{We calculate the size of a rumor outbreak using bond percolation theory.}
\item{We formulate two optimization problems viz. minimize cost subject to limit on outbreak probability and minimize outbreak probability subject to a budget on the cost. The two problems turn out to be non linear programs and we show that they can be solved by solving a set of linear programs.}
\item{We solve the problems for  sub-linear, linear and a super linear cost structure. For the linear case the solution provides useful insights such as: when more resources are available it is best to increase the quantity of trained individuals than increasing the quality of training. Also, the optimum set of trained individuals follows a pattern which can be easily converted to an implementable heuristic.}
\end{itemize}     
\par
The article is organized as follows. The model is described in Sec. \ref{sec:model}, the analysis in Sec. \ref{sec:analysis}, the main results in Sec. \ref{sec:results} and the conclusions in Sec. \ref{sec:conclusion}.
   
\section{Model \label{sec:model}}
 We divide the total population of $N$ individuals into two types: the untrained (type $1$) and the trained (type $2$). These individuals are connected with one another through a social network, which is represented by an undirected graph (network). Nodes represent individuals while a link embodies the communication pathways between individuals. For the sake of analytical tractability we make the following approximation. Instead of analyzing the adjacency matrix of the network, we obtain statistical information about the social network by calculating its degree distribution (probability that a randomly chosen node has $k$ neighbors). We then generate a synthetic network with the obtained degree distribution using the configuration model procedure \cite{molloy1995}. A sequence of $N$ integers, called the degree sequence, is obtained by sampling the degree distribution. Thus each node is associated with an integer which is assumed to be the number of half edges or stubs associated with the node. Assuming that the total number of stubs is even, each stub is chosen at random and joined with another randomly selected stub. The process continues until all stubs are exhausted. Self loops and multiple edges are possible, but the number of such self loops and multiple edges goes to zero as $N \to \infty$ with high probability. The network obtained by this process is termed as configuration model.  We assume that $N$ is large but finite.
 \par
  Let $P(k)$ be the degree distribution of the social network. Individuals can receive a messages which may be benign or malicious (rumors). We assume that the rumor deals with a single and specific topic, and hence we represent  it as a data message, i.e., rumor message $R$.  We propose an approach where the social planners train individuals and provide them resources such as access to real time photography from security cams, or access to opinions from a group of experts,  for helping them recognize rumors from true information. We assume that all trained individuals receive the same quality of training and resources. Let $\phi (k)$ be the proportion of individuals with $k$ degrees recruited for training. 
\par 
 We assume that initially a randomly chosen individual acts as a seed and transmits the $R$ message to \emph{all} its neighbors with probability $1$.  However, the individual who receives this message cannot recognize with certainty if it is a rumor, and hence she makes a guess about the nature of the message. She transmits the message to all her neighbors only if she believes that it is not a rumor. If she thinks that it is a rumor then she does not transmit it to any of her neighbors. An individual who receives the message and decides to spread it to her neighbors, does so only once and such a person is termed as a \emph{believer}. An individual who has not received the message is termed as a \emph{nonbeliever}. This models the scenario where an individual may hear about an event and report it to all her friends based on her gut feeling about the nature of the message. A believer cannot revert back to being a nonbeliever, while a nonbeliever who receives a message and concludes its a rumor is assumed to remain a nonbeliever. We call he link connecting a believer to an untrained node as a type $1$ link, and the link connecting a believer to a trained node as a type $2$ link.
 \par
 Let $H_0$ be the hypothesis  that the received message is not $R$, while $H_1$ be the hypothesis that received message is $R$.  Thus an individual who receives a rumor message transmits it with the probability  $Pr \{H_0 |$it is a rumor message$ \}$, i.e., it is the probability of the event that she fails to identify the true nature of the message. Let $T_1$ and $T_2$ be the probability of misclassifying the rumor as information for untrained and trained individuals respectively. We refer to $T_1$ as the force of rumor. Due to the training $T_2 < T_1$. However, no amount of training can accurately identify the nature of the message, hence $T_1, \ T_2 \ \in \ (0,1)$.
 \par
 We present the following scenario as an illustrative example. An untrained individual $A$ receives a rumor message, she makes an error in identifying it as a rumor, and hence passes on this message to all her connections. This event happens with probability $T_1$. A trained individual  $B$, receives this message from $A$, but unlike $A$ she correctly identifies the message as a rumor and hence does not send it to her connections. This event happens with probability $1 - T_2$. The proportion of believers, after the process terminates, is termed as size of the outbreak. This model is also termed as the independent cascade model \cite{Kempe2003}. The process can also be viewed as a rumor propagating through links when they are active. Thus a type $1 $ link is active with probability $T_1$ while a type $2$ link is active with probability $T_2$. These  are also called bond occupation probabilities \cite{Newman2002}.  
 
 \section{Analysis \label{sec:analysis}}
 We use bond percolation theory to analyze the rumor spreading process.  Let $P(k' \mid k)$ be the probability of encountering a node of degree $k'$ by traversing a randomly chosen link from a node of degree $k$. In other words, $P(k' \mid k)$ is the probability that a node with degree $k$ has a neighbor with degree $k'$. For a network generated by configuration model, $P(k' \mid k)= k'P(k')/<k>$ \cite{Newman2010}, which is independent of $k$, where $<k^i>$ is the $i^{th}$ moment of $P(k)$.
 \par
  Let $q$ be the probability of encountering a trained node by traversing a randomly chosen link from a node of degree $k$. Therefore, $q = \sum\limits_{k'=1}^{\infty}Pr($Neighboring node is trained $\mid$ neighboring node has degree $k')\cdot Pr($Neighboring node has degree $k'\mid$ original node has degree $k)$. 
 \begin{align*}
 q= \frac{1}{<k>}\sum_{k=1}^{\infty}k\phi(k)P(k)
 \end{align*}
 \par
 The probability that a randomly chosen node has $k_1$ untrained and $k_2$ trained neighbors $= P(k_1,k_2) = \sum\limits_{k:k=k_1+k_2}^{\infty} Pr(k_1,k_2\mid$node has degree $k)P(k) $. 
 \par
 \vspace{1pt}
 The event that a given node has degree $k$, is independent of the event that another node, having a common neighbor with the given node, has degree $k'$. This is true since the degree sequence is generated by independent samples from the distribution. The probability that a node is trained, is a function of its degree, hence the event that $A$ is trained (untrained) is independent of the event that $B$ is trained (untrained). This allows us to write:   
 \par
 \begin{align*}
  P(k_1,k_2)=  {k_1+k_2 \choose k_2}q^{k_2}(1-q)^{k_1}P(k_1+k_2) 
 \end{align*}
\par
 Let $Q(k)$ be the excess degree distribution, i.e., the degree distribution of a node arrived at by following a randomly chosen link without counting that link. For the configuration model $Q(k) = (k+1)P(k+1)/ <k>$. Let $Q(k_1,k_2)$ be the excess degree distribution for connections to untrained and trained nodes.
  \begin{align*}
  Q(k_1,k_2) = {k_1+k_2 \choose k_2}q^{k_2}(1-q)^{k_1}Q(k_1+k_2)
  \end{align*}
 \par
 Let $\tilde{P}(\tilde{k}_1,\tilde{k}_2)$ and $\tilde{Q}(\tilde{k}_1,\tilde{k}_2)$ be the distribution and the excess distribution of the number of untrained and trained neighbors who are believers. In other words the distribution and the excess distribution of type $i$ occupied links.
\par
{ \small
 \begin{align*}
 \tilde{P}(\tilde{k}_1,\tilde{k}_2) &= \sum_{k_1=\tilde{k}_1}^{\infty}\sum_{k_2=\tilde{k}_2}^{\infty}P(k_1,k_2)\prod_{i=1}^{2}{k_i \choose \tilde{k}_i}T_i^{\tilde{k}_i}(1-T_i)^{k_i-\tilde{k}_i} \\
 \tilde{Q}(\tilde{k}_1,\tilde{k}_2) &= \sum_{k_1=\tilde{k}_1}^{\infty}\sum_{k_2=\tilde{k}_2}^{\infty}Q(k_1,k_2)\prod_{i=1}^{2}{k_i \choose \tilde{k}_i}T_i^{\tilde{k}_i}(1-T_i)^{k_i-\tilde{k}_i}
 \end{align*}
 }
 \par
      \begin{table}[!t]
      \centering
        \begin{tabular}{l  l}
   \hline
   Generating function & Distribution \\
    \hline
   $G(u_1,u_2)$ & $P(k_1,k_2)$ \\
    $F(u_1,u_2)$ & $Q(k_1,k_2)$ \\
       $\tilde{G}(u_1,u_2)$ & $\tilde{P}(\tilde{k}_1,\tilde{k}_2)$ \\
        $\tilde{F}(u_1,u_2)$ & $\tilde{Q}(\tilde{k}_1,\tilde{k}_2)$ \\
   $\tilde{H}_i(u_1,u_2)$ & No. of untrained and trained nodes, \\ 
   & who are believers, in a component \\   
   & reached from a type $i$ link.\\
       $\tilde{J}_i(u_1,u_2)$ &  No. of untrained and trained nodes \\  
      & who are believers, in a component \\ 
         & reached from a node $i$.\\
          $\tilde{J}(u_1,u_2)$ &  No. of untrained and trained nodes \\  
   & who are believers, in a component \\ 
      & reached from a randomly chosen node.\\
     \hline
     \end{tabular}
     \caption{List of probability generating functions.}
     \label{table:pgf}
           \end{table}    
 The probability generating functions, $G(u) = \sum\limits_{k=0}^{\infty} u^kP(k)$, for the distributions discussed above are listed in Table \ref{table:pgf}.
 \par
 Now, $\tilde{G}(u_1,u_2)$ is given by
  \begin{align*}
  & \sum_{\tilde{k}_1,\tilde{k}_2}^{\infty}u_1^{\tilde{k}_1}u_2^{\tilde{k}_2}\sum_{k_1 = \tilde{k}_1}\sum_{k_2=\tilde{k}_2}P(k_1,k_2)\prod_{i=1}^{2}{k_i \choose \tilde{k}_i}T_i^{\tilde{k}_i}(1-T_i)^{k_i-\tilde{k}_i} \\
  &= \sum_{k_1,k_2}^{\infty}(1+(u_1-1)T_1)^{k_1}(1+(u_2-1)T_2)^{k_2}P(k_1,k_2) \\
  &= G\left(1+(u_1-1)T_1,1+(u_2-1)T_2\right)
 \end{align*}
 Similarly, $\tilde{F}(u_1,u_2) = F(1+(u_1-1)T_1,1+(u_2-1)T_2)$ 
 \par
A component is a \emph{small} cluster of nodes who have become believers due to rumor propagation. By small we mean that the cluster is finite and does not scale with the network size. However, at the phase transition the average size of the cluster diverges (as $N \to \infty$). A rumor outbreak is possible only when the phase transition threshold is crossed and the average size of the cluster diverges. In this regime the component is termed as a giant connected component (GCC) and it grows with the network size. Let $\tilde{H}_i(u_1,u_2)$ be the generating function of the distribution of the number of untrained and trained nodes in a component which is arrived at from a type $i$ link, i.e., edge whose end is type $i$ node. The component must contain at least one node of type $1$ or type $2$. Let  $\tilde{J}_i(u_1,u_2)$ and $\tilde{J}(u_1,u_2)$ be the generating functions of the distribution of the number of untrained and trained nodes in a component arrived at from node $i$ and a randomly chosen node respectively. 
 \begin{figure}
 \centering
  \includegraphics[width = 0.45\textwidth]{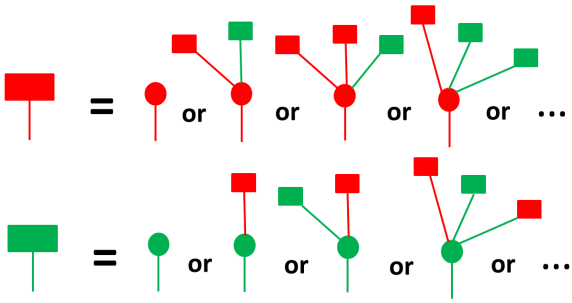}
 \caption{(Color Online) Illustration of components. The red boxes represent the components reached by type $1$ link while green boxes represent components reached by a type $2$ link. A trained node is represented by a green circle while a red circle represents an untrained node.  }
 \label{fig:percolation}
 \end{figure}

\par
The probability of encountering closed loops in finite cluster is $O(N^{-1})$ \cite{Newman2002} which can be neglected for large $N$. The tree like structure of the cluster allows us to write the size of the component encountered by traversing a link as the sum of the size of components encountered after traversing the links of the end node. This is illustrated in Fig. \ref{fig:percolation}. Since, the size of the components along different links are mutually independent (absence of loops) we can write the following equation. 
\begin{align*}
\tilde{H}_i(u_1,u_2) &= u_i \sum_{k_1,k_2}^{\infty} \tilde{H}_1^{k_1}(u_1,u_2)\tilde{H}_2^{k_2}(u_1,u_2)\tilde{Q}(\tilde{k}_1,\tilde{k}_2) \\
&= u_i\tilde{F}(\tilde{H}_1(u_1,u_2) ,\tilde{H}_2(u_1,u_2) ) \\
\end{align*}
Which can also be written as 
\par
$\tilde{H}_i(u_1,u_2) =$ 
\begin{align}
u_iF\left(1+(\tilde{H}_1(u_1,u_2)-1)T_1,1+(\tilde{H}_2(u_1,u_2)-1)T_2 \right) \label{eqn:pgfCluster}
\end{align}

The multiple $u_i$ results from the end node of the link while  $\tilde{Q}(\tilde{k}_1,\tilde{k}_2)$ is the excess distribution of the occupied type $i$ links. Similarly, $\tilde{J}(u_1,u_2)$ can be expressed as :
\begin{align*}
\tilde{J}_i(u_1,u_2) &= u_i\sum_{k_1,k_2}^{\infty} \tilde{H}_1^{k_1}(u_1,u_2)\tilde{H}_2^{k_2}(u_1,u_2)\tilde{P}(\tilde{k}_1,\tilde{k}_2) \\
&= u_i\tilde{G}(\tilde{H}_1(u_1,u_2) ,\tilde{H}_2(u_1,u_2) ) \\
\tilde{J}(u_1,u_2) &= (1-p)\tilde{J}_1(u_1,u_2) +p \tilde{J}_2(u_1,u_2)
\end{align*}
where $p$ is the probability of selecting a trained node, $p = \sum \limits_{k=1}^{\infty}P(k)\phi(k)$.
\par
The following theorem  describes the phase transition conditions required for a rumor outbreak and the size of the such a rumor outbreak. 
\begin{thm}
The condition required for a small cluster to become a giant connected component is given by: $\tilde{\nu}\geq 1 $, where 
\begin{align*}
\tilde{\nu} = T_1\sum_{k_1,k_2}^{\infty}k_1Q(k_1,k_2) + T_2\sum\limits_{k_1,k_2}^{\infty}k_2Q(k_1,k_2) 
\end{align*}
 and the size of the giant connected component is given by $1-\psi$,
\begin{align*}
\psi = \sum_{k_1,k_2}^{\infty}(1+(u^*-1)T_1)^{k_1}(1+(u^*-1)T_2)^{k_2}P(k_1,k_2)
\end{align*}
where $u^*$ is the solution of the fixed point equation
\begin{align*}
u =  \sum_{k_1,k_2}^{\infty}(1+(u-1)T_1)^{k_1}(1+(u-1)T_2)^{k_2}Q(k_1,k_2)
\end{align*}
\end{thm}
\begin{proof}
Let $<s_1>$ and $<s_2>$ be the average number of untrained and trained nodes in the component. The expected number of nodes in the component, $<s>$, is given by:
\begin{align*}
<s> &= <s_1> + <s_2>  \\
&= \frac{\partial}{\partial u_1}\tilde{J}(u_1,u_2) \biggr|_{\boldsymbol{u}=1} + \frac{\partial}{\partial u_2}\tilde{J}(u_1,u_2)  \biggr|_{\boldsymbol{u}=1}
\end{align*}
After differentiating and simplifying, $<s_1>$ can be written as:
\begin{align*}
<s_1> = (1-p) + <\tilde{k}_1> \tilde{H}_1^{'}(1,1) + <\tilde{k}_2> \tilde{H}_2^{'}(1,1) 
\end{align*}
where $<\tilde{k}_i>  = \sum \limits_{k_1,k_2}^{\infty}\tilde{k}_i \tilde{P}(\tilde{k}_1,\tilde{k}_2)$  and  
\begin{align*}
\tilde{H}_i^{'}(1,1) = \frac{\partial}{\partial u_1} \tilde{H}_i(u_1,u_2) \biggr |_{\boldsymbol{u} = 1}
\end{align*}
$\tilde{H}_i^{'}(1,1)$ can be obtained by differentiating equation (\ref{eqn:pgfCluster}).
\begin{align*}
\tilde{H}_1^{'}(1,1)&= 1 + T_1 \bar{k}_1\tilde{H}_1^{'}(1,1) + T_2 \bar{k}_2\tilde{H}_2^{'}(1,1) \\
\tilde{H}_2^{'}(1,1)&= T_1 \bar{k}_1\tilde{H}_1^{'}(1,1) + T_2 \bar{k}_2\tilde{H}_2^{'}(1,1)
\end{align*}
where $\bar{k}_i = \sum \limits_{k_1,k_2}^{\infty}k_iQ(k_1,k_2)$. Solving the two simultaneous equations we obtain $\tilde{H}_1^{'}(1,1)= \frac{1-T_2\bar{k}_2}{1 - T_1\bar{k}_1 - T_2\bar{k}_2} $ and $\tilde{H}_2^{'}(1,1)= \frac{T_1\bar{k}_1}{1 - T_1\bar{k}_1 - T_2\bar{k}_2}$. Substituting in the expression for $<s_1>$ we get.
\begin{align*}
<s_1> = (1-p) + \frac{<\tilde{k}_1>(1-T_2\bar{k}_2) + <\tilde{k}_2>T_1\bar{k}_1 }{1 - T_1\bar{k}_1 - T_2\bar{k}_2}
\end{align*}
One can similarly show that:
\begin{align*}
<s_2> =  p + \frac{<\tilde{k}_1>T_2\bar{k}_2 + <\tilde{k}_2>(1-T_1\bar{k}_1) }{1 - T_1\bar{k}_1 - T_2\bar{k}_2}
\end{align*}
Therefore,
\begin{align*}
<s> =  1 + \frac{<\tilde{k}_1> + <\tilde{k}_2>}{1 - \bar{k}_1 - \bar{k}_2}
\end{align*}
Thus, when $\bar{k}_1 + \bar{k}_2 \geq 1$, $<s>$ is no longer finite, it morphs into a giant connected component, or in other words there is a rumor outbreak. 
\par
Assume that a giant connected component of exists ($\tilde{\nu} \geq 1$). For any given node let $u_i$ be the probability that one of its type $i$ links does \emph{not} lead to the giant connecting component. The probability that a randomly chosen node is not a part of the GCC  is given by
\begin{align*}
\psi &= \sum_{\tilde{k}_1,\tilde{k}_2}^{\infty}u_1^{k_1}u_2^{k_2}P(k_1,k_2) \\
\end{align*} 
Now, $u_i$ can be written as $Pr($link is not occupied $)$ + $Pr($link is occupied and the neighbor is not connected to the GCC$)$. Mathematically this can be written as : 
\begin{align*}
u_1 &=  1 - T_1 + T_1\sum_{k_1,k_2}^{\infty}u_1^{k_1}u_2^{k_2}Q(k_1,k_2)  \\
u_2 &=  1 - T_2 + T_2\sum_{k_1,k_2}^{\infty}u_1^{k_1}u_2^{k_2}Q(k_1,k_2)
\end{align*}
Simplifying we obtain, $(u_1 -1)/T_1 = (u_2 -1)/T_2 $. Let $u := (u_1 -1)/T_1 + 1$. Hence, $u_i = 1 + (u-1)T_i$. Note that  $u_i$ is lower bounded by $1-T_i$ and upper bounded by $1$ and hence $0 \leq u \leq 1$. Substituting this in above equations we obtain the desired result:
\begin{align*}
\psi = \sum_{k_1,k_2}^{\infty}(1+(u-1)T_1)^{k_1}(1+(u-1)T_2)^{k_2}P(k_1,k_2)
\end{align*}
where $u$ must satisfy
\begin{align*}
u = \sum_{k_1,k_2}^{\infty}(1+(u-1)T_1)^{k_1}(1+(u-1)T_2)^{k_2}Q(k_1,k_2)
\end{align*}
\end{proof}

The size of the outbreak can now be used for formulating the optimization problem. The size of the outbreak can also be interpreted as the probability that a randomly chosen individual is a believer,  or as the probability of a rumor outbreak \cite{Newman2002}.

 \section{Results \label{sec:results}}
 \begin{figure}
 \centering
  \includegraphics[width = 0.45\textwidth]{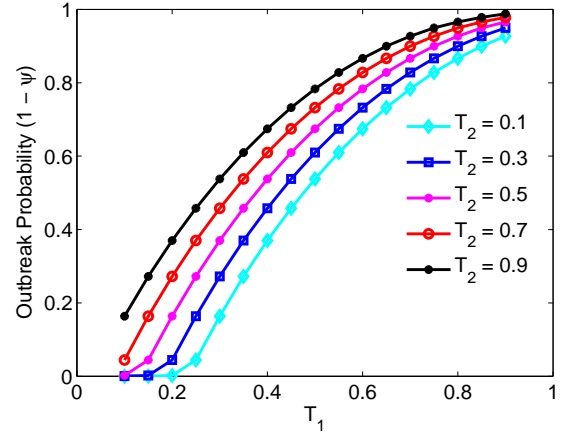}
 \caption{(Color Online) Outbreak probability  $1-\psi$ vs. $T_1$ for various values of $T_2$ }
 \label{fig:OutbreakSizeT2}
 \end{figure}
  Recruiting, training and equipping individuals with resources to distinguish between rumor and benign message is costly.  A high degree individual is more likely to receive a message than a low degree individual. For a trained individual, reception of a message translates into usage of costly resources for making a decision. Hence, we assume that the cost incurred in training an individual is an increasing function of its degree $k$. Increasing the quality of training or equivalently, decreasing $T_2$, results in decrease of the outbreak probability $1-\psi$. This is shown graphically in Fig. \ref{fig:OutbreakSizeT2} and analytically in Lemma \ref{lemma4} detailed in the Appendix. Increasing the quality of training by allocating more resources should result in a higher cost, hence the cost must be a decreasing function of $T_2$.  Let  $c(k,T_2)$ be the cost of training a node with degree $k$. The average cost is given by $\sum\limits_{k=1}^{\infty}c(k,T_2)Pr($node is selected for training $\mid$ node has degree $k)P(k) = \sum\limits_{k=1}^{\infty} c(k,T_2)\phi(k)P(k)$.  The average number of trained individuals is given by $\sum\limits_{k=1}^{\infty}\phi(k)P(k)$.
  
  We formulate two optimization problems, viz., one which minimizes cost while enforcing an upper bound on the outbreak probability, and the other which minimizes the outbreak probability for a given cost budget.  

\subsection{Cost minimization problem}
Providing guarantees on rumor outbreak at a minimum cost is appropriate in scenarios where the rumor spread may result in loss of life and property, such as rumors that incite communal violence. The guarantee on rumor outbreak probability is written as a constraint to the optimization problem. The cost $c(\boldsymbol{\phi},T_2)$ is minimized subject to $1 - \psi \leq \delta$ where $\gamma \ \in \ [0,1]$.  If $\gamma = 0$, the constraint becomes $\tilde{\nu} \leq 1$, as $\gamma = 0$ implies $\psi = 1$ which is the same as $\tilde{\nu} \leq 1$.  This is a non linear optimization problem and the non linearity arises from the outbreak probability expression.    

 For a fixed $T_2$, the constraint $1-\psi \leq \delta$ is non linear in $\boldsymbol{\phi}$. For a fixed $T_2 $, the following theorem allows us to write the non linear constraint as a linear constraint.  The proof follows from Lemmas \ref{lemma1}, \ref{lemma2} and \ref{lemma3} detailed in the Appendix.

\begin{thm} \label{theorem1}
If $T_2 <T_1$, for $\psi \ \in \  [0,1)$, then $\psi$ is strictly increasing with respect to $q$, i.e, $\frac{d\psi}{dq} >0$ for all $q \ \in \  [0,1] $ and $\tilde{\nu}$ is strictly decreasing with respect to $q$, i.e, $\frac{d\tilde{\nu}}{dq} < 0, \ \forall \ \ q \ \in \  [0,1]$, where $q = \frac{1}{<k>}\sum\limits_{k=1}^{\infty}k\phi(k)P(k)$.
\end{thm}
  Since, $\frac{d\psi}{dq} >0$, the outbreak probability constraint can be written as $\frac{1}{<k>}\sum\limits_{k=1}^{\infty}k\phi(k)P(k) \geq q^*$, where $ \psi(q) \mid_{q=q^*} \ = 1-\gamma$.
\par
The optimization problem is described by:

\begin{align}
   \begin{aligned}
   & \underset{\boldsymbol{\phi},T_2}{\text{minimize}}
    \ \ \ \ \sum_{k=1}^{\infty} c(k,T_2)\phi(k)P(k)  \\
   & \text{subject to:} \ \  \\
   &		  \frac{1}{<k>}\sum_{k=1}^{\infty}k\phi(k)P(k) \geq q^* \\
   &			\sum_{k=1}^{\infty} \phi(k)P(k) \leq B      \\
   &         T_l \leq T_2 \leq T_u \\
   &		\boldsymbol{0} \leq \boldsymbol{\phi} \leq \boldsymbol{1}  \label{eqn:LinOpt1}
     \end{aligned}
\end{align}
$B$ is the upper bound on the average number of individuals that can be trained, $B \ \in \ [0,1]$. $T_l > 0$ is the upper bound on the quality of training and $T_u < T_1$ is the lower bound. The above problem is non linear in $T_2$ as $q^*$ is a non linear function of $T_2$, however, for a fixed $T_2$ it is a linear program. We convert the problem to a set of linear programs by discretizing $T_2$ and formulating a linear program for each value of $T_2$. We obtain an approximate global minimum by comparing the minimas obtained by solving the set of linear programs. 
\par
The optimization problem described above may not be feasible for all values of $T_1$ and  for all possible degree distributions $P(k)$. Here is an example of a scenario where the constraint set is empty. Assume $B=1$, the problem  becomes infeasible if $1 - \psi \geq \gamma$ when $T_2$ is at the minimum possible value ($T_l$) and $q=1$ , i.e., all individuals are trained. However, if the constraint on the quality of training is relaxed, $T_l = 0$, then an optimal feasible solution will always exists as $T_2$ can be pushed arbitrarily close to $0$ to ensure that the outbreak probability constraint is not violated.
\par
\subsubsection{Linear cost}
For the sake of illustration we study the case where the cost is directly proportional to the degree and inversely proportional to $T2$, i.e., $c(k,T2) = \frac{k}{T_2}$. Assuming feasibility of the problem, we use a numerical linear programming solver to arrive at the solution. Since many social networks are scale free \cite{Java2007,Caldarelli2007}, we assume that the network degree distribution is a power law $P(k) \propto k^{-\alpha}$ with $\alpha = 2.5$ and a population size of $2000$. We use the algorithm detailed in \cite{Catanzaro2005} for generating the configuration model while the power law distribution was generated using the code supplied in \cite{Clauset2009}.

\begin{figure}
\centering
 \includegraphics[width = 0.45\textwidth]{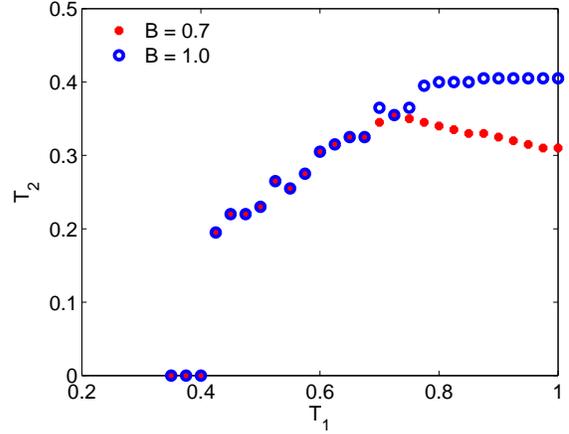}
\caption{(Color Online) Optimum value of $T_2$  vs. $T_1$. Parameters: $\gamma = 0.1 \ T_l  = 0, \ T_u = T_1$ }
\label{fig:OptT2}
\end{figure}

\begin{figure}
\centering
 \includegraphics[width = 0.45\textwidth]{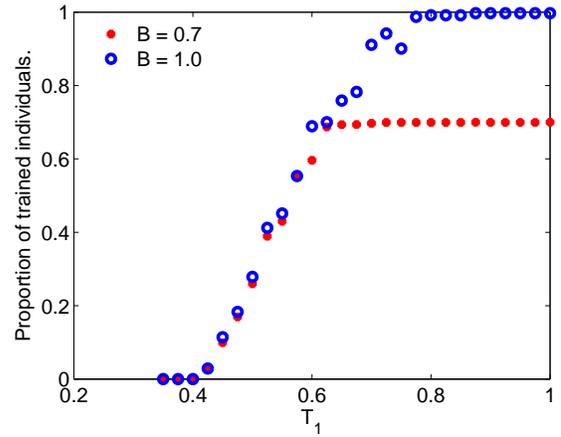}
\caption{(Color Online)  Proportion of Trained Individuals vs. $T_1$. Parameters: $\gamma = 0.1, \ T_l  = 0, \ T_u = T_1$}
\label{fig:OptNumTrained}
\end{figure}
The optimum value of $T_2$ for varying force of rumor, $T_1$, is shown in Fig. \ref{fig:OptT2}, while Fig. \ref{fig:OptNumTrained} shows the optimum proportion of trained individuals for varying $T_1$. For the $B=1$ scenario, as the force of rumor, $T_1$, increases, the optimum $T_2$ rises and saturates. Notice that when $T_2$ saturates the proportion of trained nodes hits $1$, i.e., all nodes are trained. Since all the nodes are trained, increase in $T_1$ does not have any effect on $T_2$. This behavior is observed because the cost of increasing the quality (reducing $T_2$) is higher than the cost of training individuals. 
\par
When $B=0.7$, the proportion of trained nodes cannot exceed $0.7$, and hence after the proportion of individuals hit $0.7$, increasing the quality of training (reducing $T_2$) is the only available choice to combat the increase in $T_1$.  When $T_1$ is small ($\leq 0.4$), $\boldsymbol{\phi} = \boldsymbol{0}$, thus the cost is $0$, and hence $T_2$ can take any arbitrary value.

\begin{figure}
\centering
 \includegraphics[width = 0.45\textwidth]{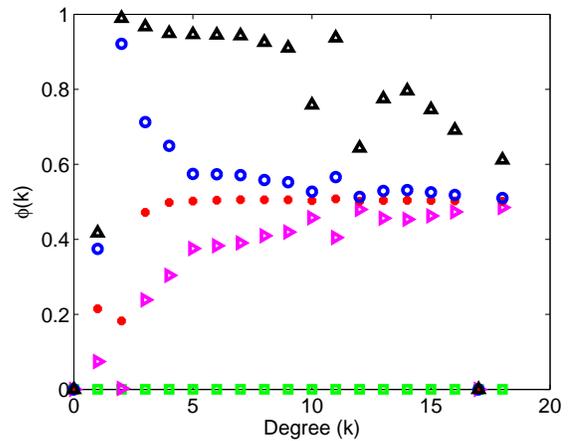}
\caption{(Color Online)  $\boldsymbol{\phi}$ for various values of $T_1$. Parameters :  $\gamma = 0.1, \ B = 0.7 \ T_l = 0, \ T_u = T_1$. Green squares : $T_1 = 0.4$, magenta triangles : $T_1 = 0.45$, red stars : $T_1 = 0.5$, blue circles : $T_1 = 0.6$, black hats :    $T_1= 0.7$.  }
\label{fig:CompareRumT2}
\end{figure}
\par
Fig. \ref{fig:CompareRumT2} shows $\boldsymbol{\phi}$ for various values of $T_1$. As $T_1$ increases the proportion of trained low degree nodes increases much faster than the proportion of trained high degree nodes. A clear pattern is seen, the proportion of high degree nodes that need to be trained is more or less constant with respect to $T_1$, while the trained low degree nodes are extremely sensitive to $T_1$. Thus, one can formulate a simple policy of training a fixed proportion (say $40-60 \%$) of high degree nodes, and recruit low degree nodes depending one the estimated severity of the rumor $T_1$. In other words, after fixing the proportion of high degree individuals that are trained, if the social planners perceive that a particular rumor message can be easily identified by individuals then they need not train a whole lot of low degree nodes. Thus, for a cost function linear in $k$, low degree nodes are more important than high degree nodes. 

\subsubsection{Comparison of linear,sub-linear and super-linear cost models}
To understand the effect of cost structure on the optimum set of trained individuals, we study three cost structures, viz., a cost function sub linear in $k$, $c(k,T_2) = \frac{\sqrt{k}}{T_2}$, linear in $k$ $c(k,T_2) = \frac{k}{T_2}$ and super linear $c(k,T_2) = \frac{k^{1.5}}{T_2}$.   

\begin{figure}
\centering
 \includegraphics[width = 0.45\textwidth]{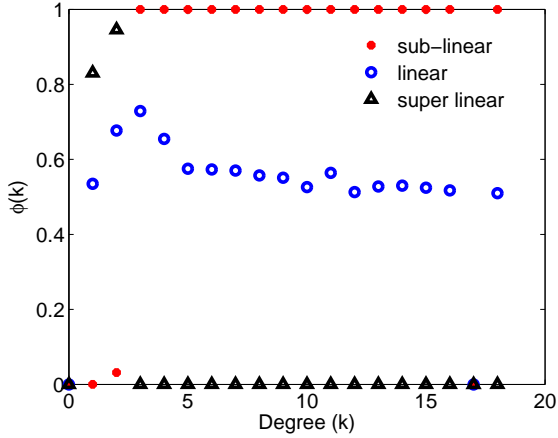}
\caption{(Color Online) $\boldsymbol{\phi}$ vs $T_1$ for various cost functions. Parameters: $\gamma = 0.1, \ T_1 = 0.7,\ T_l  = 0, \ T_u = T_1 \ B = 0.7$}
\label{fig:ComparePhiprob1}
\end{figure}

\begin{figure}
\centering
 \includegraphics[width = 0.45\textwidth]{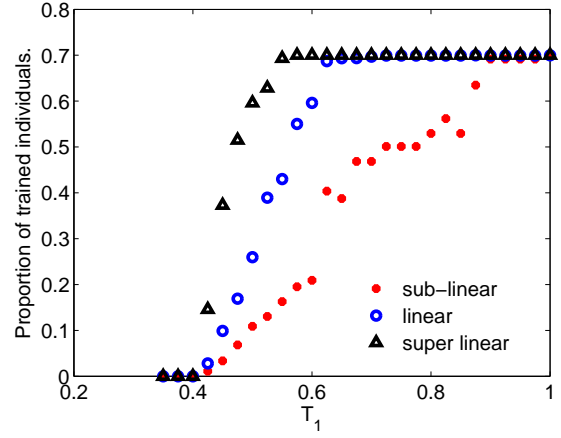}
\caption{(Color Online) Proportion of Trained Individuals vs. $T_1$. Parameters: $\gamma = 0.1, \ T_l  = 0, \ T_u = T_1, \ B = 0.7$}
\label{fig:CompareNumTrainedProb1}
\end{figure}
\par
The profile of optimum trained individuals for linear, sub linear and super linear cost functions is illustrated in Fig. \ref{fig:ComparePhiprob1}. In sub linear case only the high degree nodes are trained while in the super linear case only the low degree nodes are trained, while the linear case hits a middle ground between the sub linear and super linear cases.
\par
The reduction of outbreak probability obtained by training a high degree node is much more than that obtained by training a low degree node.  In the super linear case  since there is a significant cost on training high degree nodes, large number of low degree nodes are trained. In the sub linear case we observe the exact opposite behavior.  A scale free network has a large number of low degree nodes compared to high degree nodes, hence the number of trained individuals in the super linear case must be much higher than the sub linear case, this is can be seen in Fig. \ref{fig:CompareNumTrainedProb1}. 

\subsection{Outbreak probability minimization problem}
We now look at the problem of minimizing the outbreak probability in a resource constrained scenario. More, specifically we study a scenario where the cost budget is finite. Thus the outbreak probability $1-\psi$ must be minimized subject to a cost constraint. Like the previous optimization problem this is also a non linear program. Since $\frac{d\psi}{dq} > 0$ for a fixed $T_2$, maximizing $q$ is equivalent to minimizing $1-\psi$. Thus for a fixed $T_2$ the problem is equivalent to a linear program. We discretize $T_2$ and compute the optimum $q, \boldsymbol{\phi}$ for each $T_2$ by solving the linear program. Thus,  a set of outbreak probabilities is obtained for each optimum $q$ and the corresponding $T_2$. The approximate global minimum, optimum $\boldsymbol{\phi}$ and $T_2$ is found by obtaining the minimum outbreak probability from the set and the corresponding $q, \boldsymbol{\phi}$ and $T_2$. The optimization problem for a fixed $T_2$ is given by: 
\begin{align}
   \begin{aligned}
   & \underset{\boldsymbol{\phi}}{\text{maximize}}
    \ \ \ \ \sum_{k=1}^{\infty}k\phi(k)P(k) \\
   & \text{subject to:} \ \  \\
   &		  \sum_{k=1}^{\infty} c(k,T_2)\phi(k)P(k) \leq C \\
   &          \sum_{k=1}^{\infty} \phi(k)P(k) \leq B \\
   &		\boldsymbol{0} \leq \boldsymbol{\phi} \leq \boldsymbol{1}  \label{eqn:linOpt2}
     \end{aligned}
\end{align}

\subsubsection{Linear cost}
Like the previous problem we study the case where the cost is a linear function of $k$ and is inversely proportional to $T_2$.  We assume that the network degree distribution is a power law $P(k) \propto k^{-\alpha}$ with $\alpha = 2.5$ and a total population size of $2000$.
  \begin{figure}
  \centering
   \includegraphics[width = 0.45\textwidth]{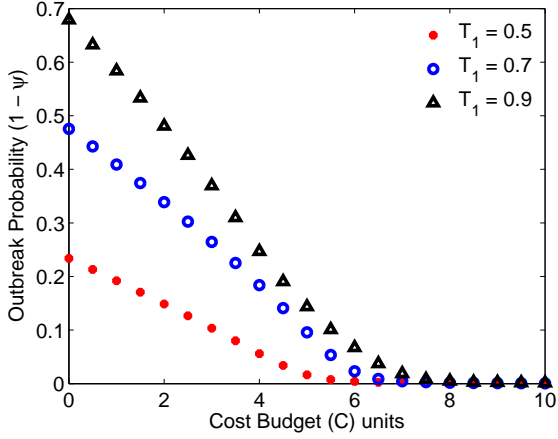}
  \caption{(Color Online) Outbreak Probability  vs. Cost Budget $C$ for various $T_1$. Parameters: $B = 0.7, \ T_l = 0,\  T_u = T_1$ }
  \label{fig:OptOutbreakBudget}
  \end{figure}
  
  \begin{figure}
  \centering
   \includegraphics[width = 0.45\textwidth]{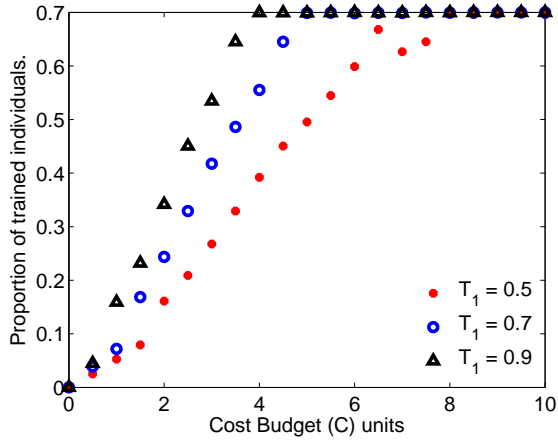}
  \caption{(Color Online) Proportion of Trained Individuals vs. Cost Budget $C$ for various $T_1$.  Parameters: $B = 0.7, \ T_l = 0,\  T_u = T_1$}
  \label{fig:OptNumTrainedBudget}
  \end{figure}
  
    \begin{figure}
    \centering
     \includegraphics[width = 0.45\textwidth]{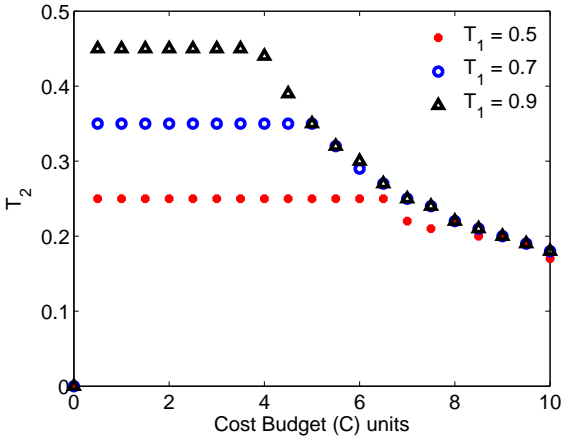}
    \caption{(Color Online) $T_2$  vs. Cost Budget $C$ for various $T_1$. Parameters: $B = 0.7, \ T_l = 0,\  T_u = T_1$}
    \label{fig:OptT2Budget}
    \end{figure}
  
  \begin{figure}
  \centering
   \includegraphics[width = 0.45\textwidth]{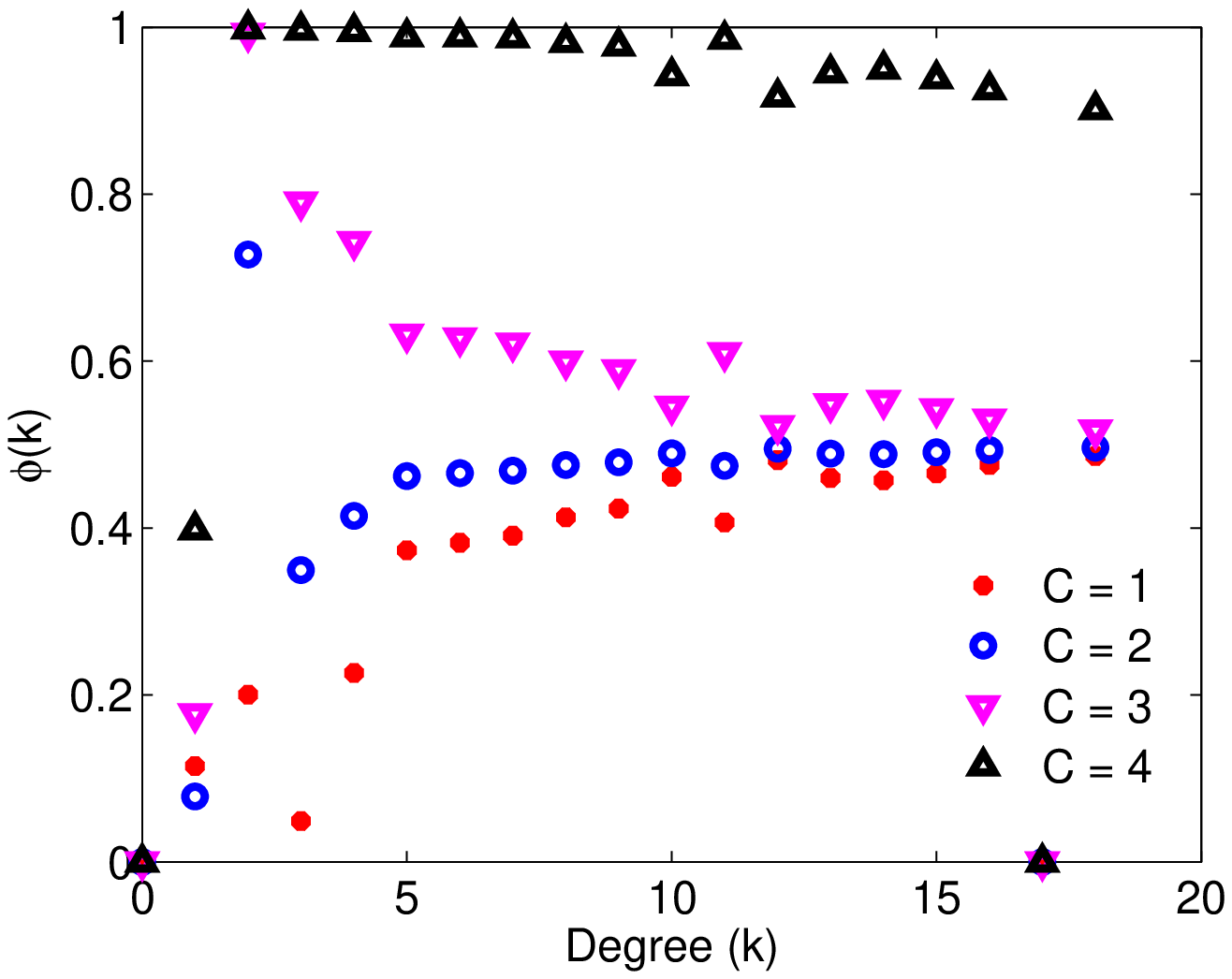}
  \caption{(Color Online) $\boldsymbol{\phi}$ for various values of cost budget $C$. Parameters :  $ B = 0.7, \ T_l = 0,\  T_u = T_1$.   }
  \label{fig:PhiBudget}
  \end{figure}
  Fig. \ref{fig:OptOutbreakBudget} illustrates the reduction in the outbreak probability with increase in cost budget $C$ for various values of $T_1$. Fig. \ref{fig:OptNumTrainedBudget} shows the expected number of trained individuals for a varying cost budget while Fig. \ref{fig:OptT2Budget} displays the corresponding optimum values of $T_2$.  When the total number of trained individuals hit the limit $B$, increasing the cost budget leads to a decrease in $T_2$. Thus, given an increasing budget, we can conclude that it is better to increase the quantity of trained individuals rather than the quality of training, and only when increasing the quantity is no longer possible, the quality should be increased. 
  \par
   In Fig. \ref{fig:PhiBudget} we plot $\boldsymbol{\phi}$ for various values of $C$. Similar to the previous optimization problem, low degree nodes are seen to be more sensitive to changes in $C$ than high degree nodes.
   
   \subsubsection{Comparison of linear,sub-linear and super-linear cost models}
   Like the previous problem, we study the effect of sub linear and super linear on the solution of the optimization problem. 
     \begin{figure}
     \centering
      \includegraphics[width = 0.45\textwidth]{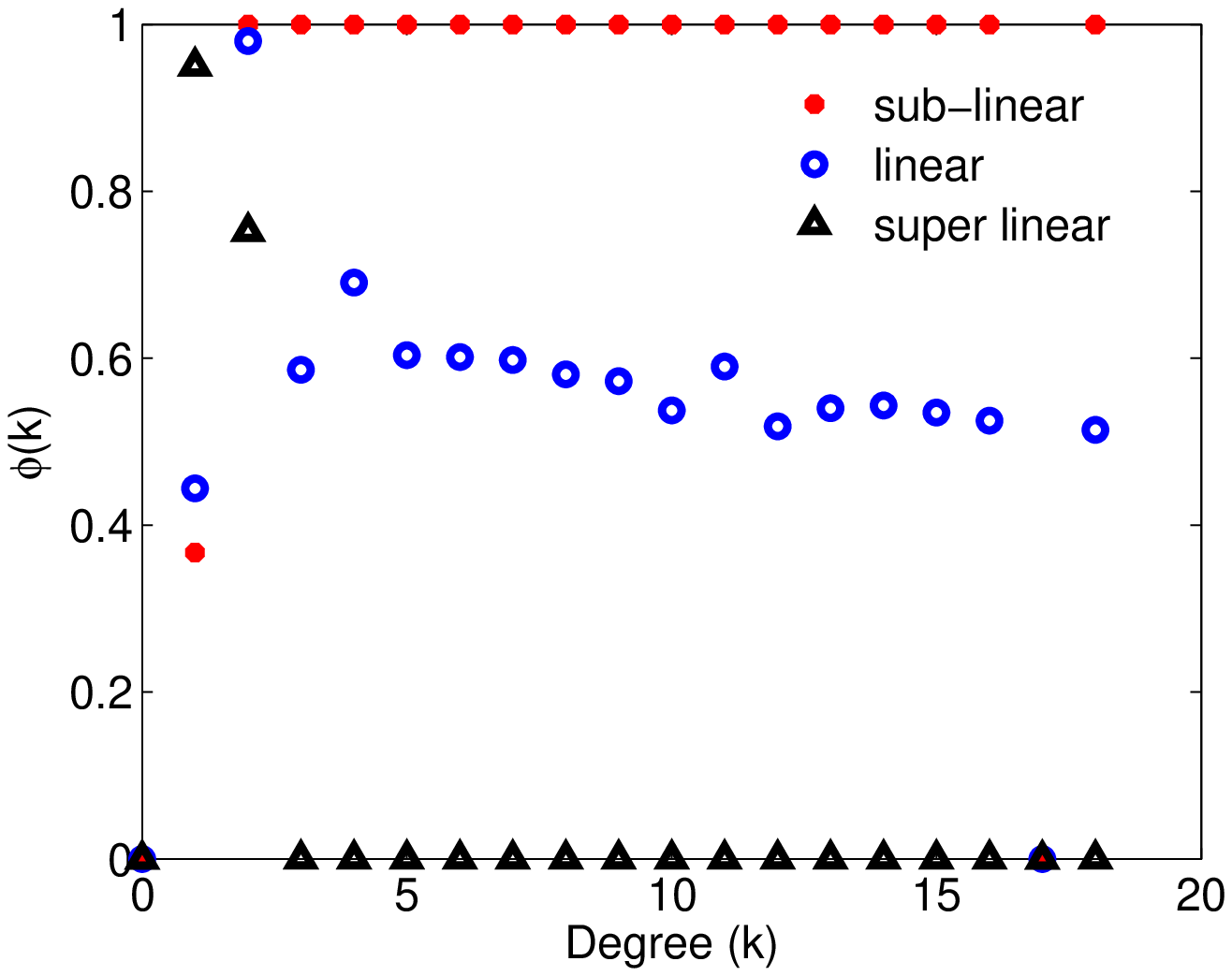}
     \caption{(Color Online) $\boldsymbol{\phi}$ for various values of cost budget $C$. Parameters :  $ B = 0.7, \ T_l = 0,\  T_u = T_1$.   }
     \label{fig:ComparePhiprob2}
     \end{figure}
    
      \begin{figure}
      \centering
      \includegraphics[width = 0.45\textwidth]{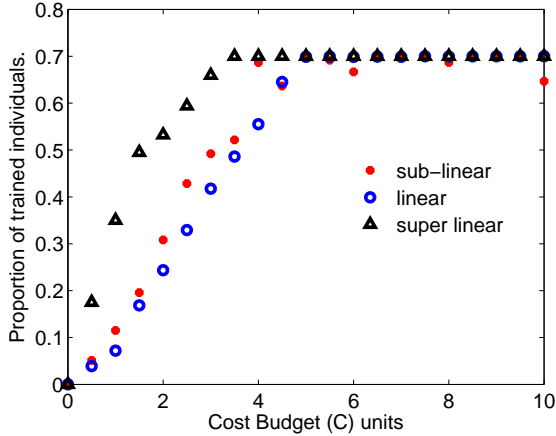}
      \caption{(Color Online) Proportion of Trained Individuals vs. $T_1$. Parameters: $B = 0.7, \ T_1 = 0.6, \ T_l  = 0, \ T_u = T_1$}
      \label{fig:CompareNumTrainedprob2}
      \end{figure}
      \begin{figure}
      \centering
       \includegraphics[width = 0.45\textwidth]{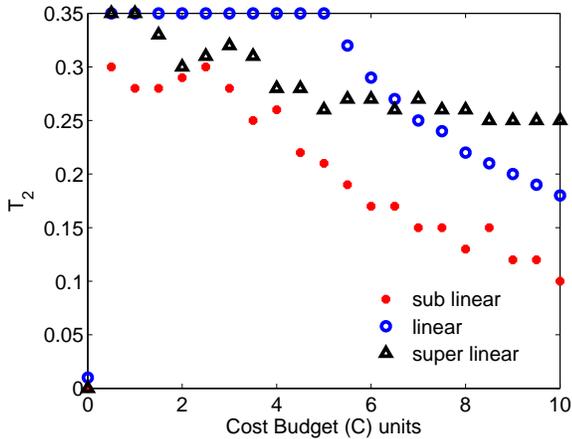}
      \caption{(Color Online) Optimum value of $T_2$  vs. $T_1$ for various cost functions. Parameters: $\ B = 0.7,  \ T_1 = 0.6, \ T_l  = 0, \ T_u = T_1$ }
      \label{fig:CompareT2prob2}
      \end{figure}
   Fig. \ref{fig:ComparePhiprob2} which plots the profile of trained individuals for varying cost budget $C$ is very similar to Fig. \ref{fig:ComparePhiprob1} in the cost minimization problem subsection, i.e., in the super linear case low degree nodes are trained while in the sub linear case high degree nodes are trained. However, the optimum $T_2$ for the super linear case displays an interesting behavior. As shown in Fig. \ref{fig:CompareNumTrainedprob2}, the optimization problem involving super linear costs hits the limit on the number of trained individuals ($B=0.7$) much before the optimization problem with sub linear costs. One would expect that, if the cost budget is increased beyond this point $T_2$ should decrease. However, Fig. \ref{fig:CompareT2prob2} shows that $T_2$ decreases very slowly for the super linear case in comparison to the sub linear case. The reason for this behavior is not entirely clear.

\subsection{A note on the information spreading problem}
 The analysis and results discussed in this article can also be applied to address the problem of spreading information in a social network. Social planners, health campaigners, or political parties may want to ensure the dissemination of information in the social network at minimal cost. Incentivizing individuals for spreading information can help in disseminating the message to a large number of people. Redefine $T_1$ as the probability an ordinary individual shares the message, and $T_2$ as the probability that an recruited individual shares the message ($T_2 > T_1$). Let $\phi(k)$ be the proportion of recruited individuals with degree $k$, and $c(k,T2)$ be the  recruitment cost. 
 \par
 An optimization problem for minimizing the cost guaranteeing the size of the information outbreak $1-\psi \geq \gamma$, or maximizing the outbreak size $1-\psi$ for a given cost budget can be formulated. Theorem \ref{theorem1} can be easily extended to include the scenario $T_2 > T_1$. In this case we would obtain $\frac{d\psi}{dq} < 0$ and $\frac{d\tilde{\nu}}{dq} > 0$. This would enable one to solve the information spreading problem by solving a set of linear programs.
        
\section{Conclusion and Future Work \label{sec:conclusion}} 
In this article we studied the problem of containing  rumors in a social network.  More specifically, we considered a scenario where individuals are unable to distinguish between the rumor and the information message, and proposed a training mechanism to recruit and train individuals. By using network percolation theory we calculated the size of the rumor outbreak and the conditions for the occurrence of such outbreaks. We then formulated an optimization problem for minimizing the expected recruitment and training cost while ensuring prevention of rumors. The optimization problem turned out to be nonlinear, and we showed that for a fixed quality of training, $T_2$, the problem becomes a linear programming problem. This enabled us to  solve the general problem  by solving a set of linear programs. As an illustration we studied the solutions of cost functions which are sub-linear, linear and super-linear in node degree $k$. For the linear case, the solution exhibited interesting properties, such as the sensitivity of low degree nodes to the intensity of rumor and the cost budget  and the lack of sensitivity of high degree nodes to the same. In the sub linear case it is best to train low degree nodes while for the super linear it is exactly the opposite: the high degree nodes. However, in both sub and super linear case the training profile was found to be largely deterministic. In all the cases the solutions exhibited  patterns which can be easily converted to an implementable heuristic. Additionally, our results can be easily extended to address the information spreading problem.
\par
More importantly, our results have implications on rumor control policies. Many Governments are concerned about dangerous rumor outbreaks and misinformation epidemics propagated on OSNs, and some, like the Indian Government and the Chinese Government have drafted policies \cite{Franke2012,chineserumor} for controlling them. However, these policies are drafted in an ad hoc manner. Furthermore they are not well studied and they do not work \cite{Franke2012}. In contrast, our approach gives guarantees on the outbreak size, uses the minimal possible resources and can be implemented in the form of a heuristic. 
\par
In this article we have assumed that only the rumor, modeled as a single message, is circulated in the social network. There may be situations where messages, other than the primary rumor message circulate in the social network. These messages may interact with the primary rumor message either helping or hindering its spread. Analysis of a model which incorporates such interactions may reveal novel approaches  for combating rumor outbreaks. We leave this interesting problem to the future.

\appendix
Theorem \ref{theorem1} is established using Lemma  \ref{lemma2} and \ref{lemma3}, while Lemma \ref{lemma1} is used in the proof of Lemma \ref{lemma2} and \ref{lemma3}.
\par
 \begin{lem} \label{lemma1}
 For all $a ,\ b \ \in \ [0,1] $ and $ k_1 + k_2 \leq n, \ n \ \in \ \boldsymbol{Z}^+$ and any arbitrary $f:\boldsymbol{Z}\rightarrow\boldsymbol{R}$ the following is true:
 
 \begin{align*}
 &\sum_{k_1=0}^{n}\sum_{k_2=0}^{n-k_1}f(k_1+k_2)k_2 {k_1+ k_2 \choose k_2}a^{k_2-1}b^{k_1} \\ -& \sum_{k_1=0}^{n}\sum_{k_2=0}^{n-k_1}f(k_1+k_2)k_1 {k_1 + k_2 \choose k_2}a^{k_2}b^{k_1-1} = 0
 \end{align*}
 
 \end{lem}
 \begin{proof}
\par
 We can change switch the indices in the second term, i.e.,
 \begin{align*}
 &\sum_{k_1=0}^{n}\sum_{k_2=0}^{n-k_1}f(k_1+k_2)k_1 {k_1 + k_2 \choose k_2}a^{k_2}b^{k_1-1}  \\
 &=\sum_{k_1=0}^{n}\sum_{k_2=0}^{n-k_1}f(k_1+k_2)k_2 {k_1 + k_2 \choose k_2}a^{k_1}b^{k_2-1}
 \end{align*}
 Hence,
  \begin{align}
  LHS = &\sum_{k_1=0}^{n}\sum_{k_2=0}^{n-k_1}  f(k_1+k_2)k_2 {k_1+ k_2 \choose k_2}\bigg{(}a^{k_2-1}b^{k_1} \nonumber \\&- a^{k_1}b^{k_2-1} \bigg{)} \nonumber \\ 
  & = \sum_{k_1=0}^{n}\sum_{k_2=0}^{n-k_1} g(k_1,k_2) \label{eqn:lemma1} 
  \end{align}
 We now count the number of terms in the above equation and show that they are even.  An expression indexed by a specific $k_1$ and $k_2$ denotes a term, e.g, $g(1,1)$ is a term. The total number of terms in the summation $= \sum\limits_{i=1}^{n+1}i = \frac{(n+1)(n+2)}{2}$. Out of those,  $n+1$ terms are $0$ due to the $k_2$ multiplier ($k_2 = 0$ for $k_1 = 0 \to n$). Additionally, when $k_2 = k_1+1$  equation (\ref{eqn:lemma1}) is zero. The total number of terms when $k_2 = k_1 +1 $ is given by $\floor*{\frac{n+1}{2}}$.
 \par
 Since, these terms are zero, subtracting out these terms from the total number of terms results in 
 \begin{align*}
 &\frac{(n+1)(n+2)}{2} - (n+1) - \floor*{\frac{n+1}{2}} \\
 &= \frac{n^2}{2} \ \text{for }n \text{ even}\\
 & = \frac{(n-1)(n+1)}{2} \ \text{for }n \text{ odd}
 \end{align*}
 
Thus, the remaining terms are even for both $n$ odd and even. This allows us to pair the terms.  Consider one such pairing: the term with indices $k_1, \ k_2$ are paired with a term with indices $\hat{k}_1, \ \hat{k}_2$ where $\hat{k}_2 = k_1 + 1$ and $\hat{k}_1 = k_2 -1$.  If we sum these two terms we obtain
 \par
{\footnotesize
 \begin{align*}
 & g(k_1,k_2) + g(\hat{k}_1,\hat{k}_2) \\
 & = f(k_1+k_2)a^{k_2-1}b^{k_1}\left(  \frac{k_2(k_1+k_2)!}{k_1!k_2!} - \frac{k_2(k_1+k_2)!}{k_1!k_2!}  \right)\\&+  f(k_1+k_2)a^{k_1}b^{k_2-1}\left( \frac{(k_1+1)(k_1+k_2)!}{(k_2-1)!(k_1+1)!} - \frac{(k_1+1)(k_1+k_2)!}{(k_2-1)!(k_1+1)!} \right)  \\
 &=0
 \end{align*}
 }
 \par
 Thus, the summation of the remaining terms is zero, which completes the proof. 
  \end{proof}
  
  \begin{lem}\label{lemma2}
  If $T_2 <T_1$ then $\tilde{\nu}$ is strictly decreasing with respect to $q$, i.e, $\frac{d\tilde{\nu}}{dq} < 0, \ \forall \ q \ \in \ [0,1]$.
  \end{lem}
  \begin{proof}
\par
  $\frac{d}{dq}\tilde{\nu} =$ 
  \par
  {\small
  \begin{align*}
  & T_1\sum_{k_1,k_2}^{\infty}k_1Q(k_1+k_2){k_1+k_2 \choose k_2}\left(k_2q^{k_2-1}r^{k_1} - k_1q^{k_2}r^{k_1-1}\right) \\
  &+ T_2\sum_{k_1,k_2}^{\infty}k_2Q(k_1+k_2){k_1+k_2 \choose k_2}\left(k_2q^{k_2-1}r^{k_1} - k_1q^{k_2}r^{k_1-1}\right)
  \end{align*}
  }
  \par
  where $r = 1-q$. Let,
  \par
  {\small
  \begin{align*}
  a_1 &= \sum_{k_1,k_2}^{\infty}k_1Q(k_1+k_2){k_1+k_2 \choose k_2}\left(k_2q^{k_2-1}r^{k_1} - k_1q^{k_2}r^{k_1-1}\right) \\ 
  a_2 &= \sum_{k_1,k_2}^{\infty}k_2Q(k_1+k_2){k_1+k_2 \choose k_2}\left(k_2q^{k_2-1}r^{k_1} - k_1q^{k_2}r^{k_1-1}\right)
  \end{align*}
  }
  \par
  Adding $a_1$ and $a_2$ we get
  \par
  {\small
  \begin{align*}
  & a_1 + a_2 = \\ & \sum_{k_1,k_2}^{\infty}(k_1+k_2)Q(k_1+k_2){k_1+k_2 \choose k_2}\left(k_2q^{k_2-1}r^{k_1} - k_1q^{k_2}r^{k_1-1}\right)
  \end{align*}
  }
  \par
  Real world networks are always finite, hence $Q(k_1+k_2)=0$ for $k_1+k_2 > k_{max}$, where $k_{max}$ is the maximum degree. From Lemma \ref{lemma1}, $a_1 + a_2 = 0$. Now we prove that $a_2 > 0$. Let $k_1 + k_2 = m$. 
  
  \begin{align*}
  a_2 = \sum_{m=1}^{k_{max}}Q(m)\bigg[ &\frac{1}{q}\sum_{k_2=0}^{m}k_2^2{m \choose k_2}q^{k_2}r^{m-k_2}  \\ &-\frac{1}{r}\sum_{k_1=0}^{m}k_1(m-k_1){m \choose k_1}q^{m-k_1}r^{k_1}\bigg ]
  \end{align*}
  The summations are the second moments of a binomial random variable. $E \left[X^2 \right] = Var\left[X\right] + E \left[X \right]^2$, $E[X] = mq, \ Var[X] = mqr$.
  \par
  {\small
  \begin{align*}
  a_2 &= \sum_{m=1}^{k_{max}}Q(m)\bigg[\frac{1}{q}(mqr + m^2q^2)  -\frac{1}{r}(m^2r - mqr-m^2r^2)\bigg ] \\
  &= \sum_{m=1}^{k_{max}}Q(m)m \\
  &>0
  \end{align*}
  }
  \par
  Since $T_2<T_1$, $T_1a_1 + T_2a_2 < 0$, which completes the proof.
  \end{proof}
  
  \begin{lem} \label{lemma3}
  For $\psi \ \in \  (0,1)$, if $T_2 <T_1$ then $\psi$ is strictly increasing with respect to $q$, i.e, $\frac{d\psi}{dq} >0,  \ \forall \ q \ \in \ [0,1]$. 
  \end{lem}
  \begin{proof}
  Let, $\psi = g(u^*,q)$ where $u^*$ is the solution of the fixed point equation $u = f(u,q)$.
  \par
  {\small
  \begin{align*}
  g(u^*,q) &= \sum_{k_1,k_2}^{\infty}\alpha^{k_1}\beta^{k_2}P(k_1+k_2)  {k_1+k_2 \choose k_2}q^{k2}(1-q)^{k_1} \\
  f(u,q)&=\sum_{k_1,k_2}^{\infty}\alpha^{k_1}\beta^{k_2}Q(k_1+k_2){k_1+k_2 \choose k_2}q^{k2}(1-q)^{k_1}
  \end{align*}
  }
  \par
 where $\alpha = 1+(u^*-1)T_1$ and $\beta = 1+(u^*-1)T_2$. We first show that the solution to the fixed point equation is strictly increasing with $q$ . It can be easily shown that $\frac{\partial f(u,q)}{\partial u} > 0 $ and $\frac{\partial^2 f(u,q)}{\partial u^2} > 0$ for all $u, \ q \ \in [0,1]$. Thus $f$ is a convex function in $u$ for any fixed $q$. Also $f(0,q) > 0$ for all $q \ \in [0,1]$. 

  $\frac{\partial f(u,q)}{\partial q} =$
  \par
  {\small
  \begin{align*}
   &\sum_{k_1,k_2}^{\infty}\alpha^{k_1}\beta^{k_2}Q(k_1+k_2){k_1+k_2 \choose k_2}\left(k_2q^{k2-1}r^{k_1}-k_1q^{k2}r^{k_1-1}\right) \\
   &=\beta\sum_{k_1,k_2}^{\infty}Q(k_1+k_2){k_1+k_2 \choose k_2}k_2(\beta q)^{k2-1}(\alpha r)^{k_1} \\
   &-\alpha\sum_{k_1,k_2}^{\infty}Q(k_1+k_2){k_1+k_2 \choose k_2}k_1(\beta q)^{k_2}(\alpha r)^{k_1-1} 
  \end{align*}
  }
  \par
  From Lemma \ref{lemma1},  
  \begin{align*}
  &\sum_{k_1,k_2}^{\infty}Q(k_1+k_2){k_1+k_2 \choose k_2}k_2(\beta q)^{k2-1}(\alpha r)^{k_1} \\
   &-\sum_{k_1,k_2}^{\infty}Q(k_1+k_2){k_1+k_2 \choose k_2}k_1(\beta q)^{k_2}(\alpha r)^{k_1-1}  \\
   & = 0
  \end{align*}
  \par
  \begin{figure}
   \includegraphics[width = 0.45\textwidth]{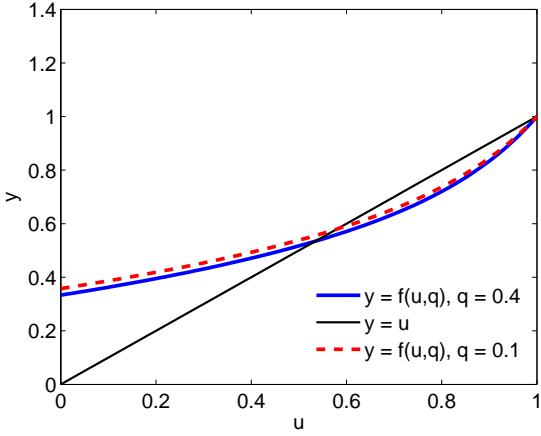}
  \caption{(Color Online) Fixed point equation, $T_1 = 0.7, \ T_2 = 0.1$}
  \label{fig:fpteqn}
  \end{figure}
  \par
    Now $\beta > \alpha$ because $T_2 < T_1$, which implies  $\frac{\partial f(u,q)}{\partial q} >0$.  
  \par
  Let $u_0^*$ be the solution of the fixed point equation for some $q=q_0$ and let $u_1^*$ be the solution to the fixed point equation when $q=q_1$ where $q_0 < q_1$. Also, $u_1^*$ exists since we have assumed $\psi < 1$. 
  \par
   The curve $y = f(u,q_1)$ lies above the curve $y = f(u,q_0)$ and hence, $u_0^* < f(u_0^*,q_1) $, or in other words, the curve $y = f(u,q_1)$ has shifted above the original fixed point. Consider set $I$ of points $u$ such that $f(u,q_1) >  f(u_0^*,q_1)$, clearly $u>u_0^*$ for all $u \ \in \ I$  (as $\frac{\partial f}{\partial u} > 0$). The line $y=u, \ \forall  \ u \in \ I$ lies above  the curve $y=f(u,q_0)$ for $u \geq u_0^*$ because $f(u)$ is convex in $u$ and the point $u=1$ is a fixed point (the line cannot be a tangent).  Since, $f(u,q_1) > f(u,q_0), \ \forall \ u \in \ [0,1) $, there must exist a point $u_1^*$ belonging to set $I$ which lies on the line $y=u$, i.e., $u_1^* = f(u_1^*,q_2)$.  Since $f(u,q)$ is continuous and differentiable in $u, q$ $\frac{du^*}{dq} > 0$. This is illustrated in Fig. \ref{fig:fpteqn}.
  
  The function $g(u^*,q)$ has the same structure as the function $f(u,q)$, and hence using the same procedure it can be shown that $\frac{\partial g(u,q)}{\partial q} > 0$. The total derivative  $\frac{d\psi}{dq}$ is given by:
  \begin{align*}
  \frac{d\psi}{dq} = \frac{\partial g}{\partial q} + \frac{\partial g}{\partial u}\frac{du^*}{dq}
  \end{align*}   
  Since all the terms on the right hand side of the above equation are positive, $\frac{d\psi}{dq} > 0$. 
  \end{proof}
  
\vspace{.1in}
  \par
  Theorem \ref{theorem1} follows from Lemma \ref{lemma2} and \ref{lemma3}.
  
  \begin{lem} \label{lemma4}
  For a fixed $T_1$ and $q$,  $\frac{d \psi}{dT_2} < 0 $ for $\psi \ \in \ [0,1)$.
  \end{lem}
  \begin{proof}
  Now, $\psi = g(u^*,T_2)$, $u^*$ is the solution to the fixed point equation $u = f(u,T_2)$ where $f$ and $g$ are given by
  \par
  {\small
  \begin{align*}
  g(u^*,T_2) &= \sum_{k_1,k_2}^{\infty}\alpha^{k_1}\beta^{k_2}P(k_1+k_2)  {k_1+k_2 \choose k_2}q^{k2}(1-q)^{k_1} \\
  f(u,T_2)&=\sum_{k_1,k_2}^{\infty}\alpha^{k_1}\beta^{k_2}Q(k_1+k_2){k_1+k_2 \choose k_2}q^{k2}(1-q)^{k_1}
  \end{align*}
  }
  \par
  where $\alpha = 1+(u^*-1)T_1$ and $\beta = 1+(u^*-1)T_2$. 
  \par
  It is easy to show that $\frac{\partial f}{\partial T_2} < 0$ and $\frac{\partial g}{\partial T_2} < 0$. Using arguments similar to the ones described in Lemma \ref{lemma3} one can write, $\frac{du^*}{dT_2} < 0$.
The total derivative  $\frac{d\psi}{dT_2}$ is given by: 
    \begin{align*}
    \frac{d\psi}{dT_2} = \frac{\partial g}{\partial T_2} + \frac{\partial g}{\partial u}\frac{du^*}{dT_2}
    \end{align*}   
    Since $\frac{\partial g}{\partial T_2}<0$, $\frac{du^*}{dT_2}<0$ and $ \frac{\partial g}{\partial u} >0 $, we have $\frac{d\psi}{dT_2} < 0$.
  \end{proof}

\bibliographystyle{IEEEtran}
\bibliography{Information_Rumor_Spread}

\end{document}